\documentclass[leqno]{amsart}

\usepackage[numbers]{tsnatbib}
\usepackage{mathabx}
\usepackage{amsmath, amsthm}
\usepackage{amsfonts,amssymb,dsfont}
\usepackage{nicefrac,booktabs,mathrsfs}
\usepackage{bm}                                 
\usepackage{stmaryrd}    

\newcommand{\cA}{{\mathcal A}}  

\newcommand{\cB}{{\mathcal B}}
\newcommand{\cO}{{\mathcal O}}

\newcommand{\ccF}{{\mathscr F}}\newcommand{\cF}{{\mathcal F}}
\newcommand{\ccG}{{\mathscr G}}\newcommand{\cG}{{\mathcal G}}
\newcommand{\ccH}{{\mathscr H}}\newcommand{\cH}{{\mathcal H}}

\newcommand{\cN}{{\mathcal N}}

\newcommand{\cU}{{\mathcal U}}

\newcommand{\cX}{{\mathcal X}}

\DeclareMathOperator{\Var}{Var}

\renewenvironment{description}
  {\list{}{\leftmargin 18pt \itemindent 14pt }}
  {\endlist}

\newenvironment{enumeratei}
  {\begin{enumerate} }
  {\end{enumerate}}

\newcommand{\half}{\frac{1}{2}}

\newcommand{\Ind}{{\mathds 1}}
\newcommand{\ind}[1]{\Ind_{\{#1\}}}

\newcommand{\R}{\mathbb{R}}

\newcommand{\bbF}{\mathbb{F}}
\newcommand{\bbH}{\mathbb{H}}
\newcommand{\bbG}{\mathbb{G}}
\newcommand{\N}{\mathbb{N}}

\newcommand{\E}{\mathbb{E}}

\newcommand{\btheta}{{\boldsymbol \theta}}

\newtheorem{theorem}{Theorem}[section]
\newtheorem{corollary}[theorem]{Corollary}      
\newtheorem{lemma}[theorem]{Lemma}              
\newtheorem{proposition}[theorem]{Proposition}  

\theoremstyle{definition}
\newtheorem{example}{Example}[section]
\newtheorem{definition}{Definition}[section]
\newtheorem{remark}{Remark}[section]

\newtheorem{assumption}{Assumption}[section]

\renewcommand{\cF}{\ccF}
\renewcommand{\cG}{\ccG}
\renewcommand{\cH}{\ccH} 

\usepackage[pdftex,colorlinks,urlcolor=red,citecolor=blue,linkcolor=red]{hyperref}

\renewcommand{\cF}{\ccF}
\renewcommand{\cG}{\ccG}
\renewcommand{\cH}{\ccH}

\begin{document}

\title[Defaultable Term Structure Modelling]{
Dynamic Defaultable Term Structure Modelling beyond the Intensity Paradigm
}
    \author {Frank Gehmlich}
    \author {Thorsten Schmidt}
    \address{University of Freiburg, Mathematical Insitute, Eckerstr. 2, 79106 Freiburg, Germany. 
             Email: thorsten.schmidt@stochastik.uni-freiburg.de. }
    \thanks{We thank Monique Jeanblanc, Shiqi Song, Ania Aksamit, Claudio Fontana, R\"udiger Frey and two anonymous referees for very valuable discussions. Financial support from the Chair Risque de Cr\'edit is gratefully acknowledged}
    \date{\today}

\begin{abstract}
The two main approaches in credit risk are the structural approach pioneered in Merton (1974) and the reduced-form framework proposed in Jarrow \& Turnbull (1995) and in Artzner \& Delbaen (1995). 
The goal of this article is to provide a unified view on both approaches. This is achieved by studying reduced-form approaches under weak assumptions. In particular we do not assume the global existence of a default intensity and allow default at fixed or predictable times with positive probability, such as coupon payment dates. 

In this generalized framework we study dynamic term structures prone to default risk following the forward-rate approach proposed in Heath-Jarrow-Morton (1992). It turns out, that previously considered models lead to arbitrage possibilities when default may happen at a predictable time with positive probability. A suitable generalization of the forward-rate approach contains an additional stochastic integral with atoms at predictable  times and  necessary and sufficient conditions for an appropriate no-arbitrage condition (NAFL) are given.
In the view of efficient implementations we develop a new class of affine models which do not satisfy the standard assumption of stochastic continuity.  

The chosen approach  is intimately related to the theory of enlargement of filtrations, to which we provide a small example by means of filtering theory  where the Az\'ema supermartingale contains upward and downward jumps, both at predictable and totally inaccessible stopping times.
\end{abstract}

\maketitle

\vspace{2mm}

\keywords{\noindent Keywords: credit risk, HJM, forward-rate, structural approach, reduced-form approach, Az\'ema supermartingale, affine processes, filtering.}

\section{Introduction}

The two most common approaches to credit risk modelling are the \emph{structural} approach, pioneered in the seminal work of Merton \cite{Merton1974}, and the \emph{reduced-form} approach which can be traced back to early works of Jarrow, Lando, and Turnbull \cite{JarrowTurnbull1995,Lando94} and to \cite{ArtznerDelbaen95}. 
In structural approaches, default happens when the company is not able to meet its obligations or a certain lower bound is hit by the value of the firm's assets. In many cases this happens when a promised payment cannot be made, which in the approaches in \cite{Merton1974} and its extensions \cite{Geske1977,GeskeJohnson84} leads to  default at  pre-specified times, such as coupon dates. The recently missed coupon payment by Argentina is an example for such a credit event as well as the default of Greece on the 1st of July\footnote{Argentina's missed coupon payment on \$29 billion debt was voted a credit event by the International Swaps and Derivatives Association, see the announcements in \cite{ISDAArgentina2014} and \cite{ReutersArgentina2014}. Regarding the failure of 1.5 Billon EUR of Greece on a scheduled debt repayment to the International Monetary fund, see e.g. \cite{NYTimes2015}.}.
The possibility that default happens at a predictable time is in strong contrast to most reduced-form approaches:  the class of so-called \emph{intensity-based} models is characterized by the property, that the default time admits an intensity and hence it avoids predictable  and in particular deterministic times. 
For a review and guide to the rich literature on this subject we refer to \cite{BieleckiRutkowski2002,LandoBook,MFE}.

Despite their theoretical appeal, structural models pose a number of challenges when used for pricing credit derivatives: first, the underlying value of the firm is difficult to observe. This fact motivated  approaches with incomplete information, see \cite{DuffieLando2001,collin-goldstein-hugonnier-04,FreySchmidt2009}. Second, the full priority hierarchy of the firm's capital structure is needed which additionally complicates the matter. Reduced-form approaches circumvent these difficulties and are less ambitious about the precise mechanism leading to default, see for example \cite{ArtznerDelbaen95,DuffieSchroderSkiadas96,JarrowLandoTurnbull97}. Closely related are dynamic term structure approaches following the pioneering work of Heath-Jarrow-Morton (HJM), \cite{HJM}, and its extension  to default risk, cf. \cite{Schoenbucher:CRDerivatives,DuffieSingleton99}. One offspring of this approach are highly tractable affine and quadratic factor models, see \cite{ChenFilipovicPoor,Duffie05,ErraisGieseckeGoldberg2010}.

It is a remarkable observation of \cite{BelangerShreveWong2004} that it is possible to extend the reduced-form approach beyond the class of intensity-based models. The authors study a class of first-passage time models under a filtration generated by a Brownian motion and show its use for pricing and modelling credit risky bonds. Our goal is to start with even weaker assumptions on the default time and to allow for jumps in the compensator of the default time at predictable times. From this general viewpoint it turns out, surprisingly, that previously used HJM approaches lead to arbitrage: the whole term structure is absolutely continuous and can not compensate for points in time bearing a positive default probability. We propose a suitable extension with an additional  term allowing for discontinuities in the term structure at  certain random times and derive precise drift conditions for an appropriate  no-arbitrage condition,  no asymptotic free lunch (NAFL).

There are approaches that bridge the gap between structural and reduced-form models, for example by considering an incomplete or noisy observation of the firm value. The first work in this direction, \cite{DuffieLando2001}, studied a first-passage time approach proposed in \cite{BlackCox1976} and showed that under incomplete information one arrives at an intensity-based model. This property can, however, not be extended to structural models where default happens at a fixed time with positive probability, like in the Merton model or its extensions. This underlines the importance of a general approach in this regard and we provide some illustrating examples. Even more, this observation motivates  a specific class of models, which we call extended Merton models, where default may happen at deterministic times with positive probability. In this class we can develop highly tractable factor models and we propose a new class of affine models, which are not stochastically continuous, matching this natural property of a generalized Merton model. 
Certainly, our approach is intimately related to the theory of enlargement of filtrations.  We provide a simple example inspired by filtering theory where the Az\'ema supermartingale contains upward and downward jumps, both at predictable and totally inaccessible stopping times.

The structure of the article is as follows: in Section \ref{sec:tau}, we introduce the general setting and study drift conditions in an extended HJM-framework which guarantee absence of arbitrage in the bond market. In Section \ref{sec:merton} we additionally assume that default may happen with positive probability only at a finite number of deterministic times and study the absence of arbitrage in this class, which we call generalized Merton models. Section \ref{sec:filtering} studies structural models under  incomplete information in  detail. Besides this, we construct  a default time where the Az\'ema supermartingale naturally has upward and downward jumps by means of filtering methods. In Section \ref{sec:doublystochastic} we give a precise construction of arbitrage-free HJM-models in a doubly-stochastic setting while Section \ref{secAffine} studies a new class of affine models which are stochastically discontinuous. Section \ref{sec:conclusion} concludes.

\section{A general account on credit risky bond markets}\label{sec:tau}
 
Consider a filtered probability space $(\Omega, \cA, \bbF, P)$ with a filtration $\bbF=(\cF_t)_{t \ge 0}$ satisfying the usual conditions, i.e.\ it is right-continuous and  $\cF_0$ contains the $P$-nullsets $N_0$ of $\cA$. Throughout, the probability measure $P$ denotes the objective measure. As we use tools from stochastic analysis, all appearing filtrations shall satisfy the usual conditions. We follow the notation from \cite{JacodShiryaev} and refer to this work for details on stochastic processes which are not laid out here.

The filtration $\bbF$ contains all available information in the market. The default of a company is public information and we therefore assume that the default time $\tau$ is an $\bbF$-stopping time.
We denote the \emph{default indicator process} $H$ by
	$$ H_t = \ind{t \ge \tau}, \qquad t \ge 0, $$
such that $H_t=\Ind_{\ldbrack\tau,\infty\ldbrack}(t)$ is a right-continuous, increasing process. We will also make use of the \emph{survival process} $1-H=\Ind_{\ldbrack 0, \tau\ldbrack}$.

A credit risky bond with maturity $T$ is a contingent claim promising to pay one unit of currency at  $T$. We denote the price of the bond with maturity $T$ at time $t \le T$ by $P(t,T)$. If no default occurred prior to or at $T$ we have that $P(T,T)=1$. 
We will consider zero recovery, i.e.\ the bond loses its total value at default, such that $P(t,T)=0$ on $\{t \ge \tau\}$. Extensions to different types of recovery can be treated along the lines of \cite{BelangerShreveWong2004}. The family of stochastic processes $\{(P(t,T)_{0 \le t \le T})$, $T\ge 0\}$ describes  the  evolution of the \emph{term structure} $T \mapsto P(.,T)$ over time.

Besides the bonds there is a \emph{num\'eraire} $X^0$, which is a strictly positive, adapted process. We assume without loss of generality that $X^0_0=1$. Moreover, we make the  weak assumption that $X^0$ is absolutely continuous. Then a \emph{short-rate}  exists, which is a progressively measurable process $r$ such that $X^0_t=\exp(\int_0^t r_s ds)$. For practical applications one would use the overnight index swap (OIS) rate for constructing such a num\'eraire.

\subsection{Absence of arbitrage in credit risky bond markets}
The market of defaultable bonds contains  an infinite number of assets and is treated here in the spirit of large financial markets following \cite{KleinSchmidtTeichmann2015}. It is the first time that this concept is applied to a market with credit risk. We therefore give a short introduction to the topic.

 Fix a finite time horizon $T^*>0$ and consider the filtration  $\bbF=(\cF_t)_{0 \le t \le T^*}$. 
We will need the following assumption on right-continuity of the bond prices in $T$ and on uniform local boundedness of bond prices. Recall that by $N_0$ we denoted the $P$-nullsets of $\cA$. For a generic process $X$ and a random time $\sigma$ we denote by $X^{(\sigma)}=(X_{t \wedge \sigma})_{t \ge 0}$  the process stopped at $\sigma$. By $a \wedge b:=\min(a,b)$ we denote the minimum of $a$ and $b$.
\begin{assumption}\label{ass1}
It holds that the set 
\begin{align*}
 \bigcup_{t \in [0,T^*]} \big\{ \omega: T \to P(t,T)(\omega) \text{ is not right-continuous}\big\} 
\end{align*}
is contained in  a $P$-nullset.
Moreover, for any $T \in [0,T^*)$ there are $\epsilon>0$, an increasing sequence of stopping times
$\sigma_n\to\infty$ and $\kappa_n\in [0,\infty)$ such that
$$P(t,U)^{(\sigma_n)}\leq \kappa_n,$$ for all $U\in[T,T+\epsilon)$ and all
$t\leq T$.
\end{assumption}

In classical HJM-models absolute continuity with respect to the maturity always holds. To our knowledge, the models studied later in this paper are the first dynamic term structure models which explicitly incorporate discontinuities in the term structure $T \mapsto P(.,T)$. In a general setting of default times as  considered here we will show that their presence is also necessary to guarantee absence of arbitrage.

\begin{definition}\label{LFM}
Fix a sequence $(T_i)_{i\in\N}$ in $[0,T^*]$.
Define the $n+1$-dimensional stochastic process  $(\mathbf S^n)=(S^{0},S^{1},\dots,S^{n})$ as follows:
\begin{align}S^{i}_t= (X_t^0)^{-1} P(t \wedge T_i ,T_i),\qquad 0 \le t \le T^*,
 \label{defSi}
\end{align}
for $i=1,\dots,n$ and $S^0_t\equiv 1$. The large financial market consists of the  sequence  $(\mathbf S^n, n \ge 1)$ of classical markets.
\end{definition}

We denote by $L^0_+$ the space of equivalence classes of measurable functions which are non-negative almost surely. By $L^1$ we denote the space of all integrable and measurable functions. Its dual space is the space of bounded measurable functions, denoted by $L^\infty$ and the bounded, non-negative measurable functions are denoted by $L_+^\infty$.  The weak-$^*$ topology on $L^\infty$ is the topology $\sigma(L^\infty,L^1)$, generated by the norm $\parallel X\parallel_*:=\sup\{E[X\varphi]:\varphi \in L^1, E[\varphi] \le 1 \}$. For further details, see for example Section A.7 in \cite{FoellmerSchied2004}. 

Absence of arbitrage is considered for each finite market $\mathbf{S}^n$ and appropriate limits.
Let $\btheta$ be a predictable $\mathbf{S}^n$-integrable process and denote by $(\btheta\cdot \mathbf{S}^n)_t$ the stochastic integral
of $\btheta$ with respect to $\mathbf{S}^n$ until $t$. The process $\btheta$ is called \emph{admissible
trading strategy} if $\btheta_0=0$ and there is an $a>0$ such that $(\btheta\cdot \mathbf{S}^n)_t\ge -a$ $P$-almost surely for all  $t \in [0,T^*]$.
Define the following cones:
\begin{equation}
\mathbf{K}^n =\{(\btheta\cdot \mathbf{S}^n)_{T^*}:\text{$\btheta$ admissible}\}\text{ and }
\mathbf{C}^n =(\mathbf{K}^n-L^0_+)\cap L^{\infty}.\label{K}
\end{equation}
$\mathbf{K}^n$ contains all replicable claims in the finite market $n$, and $\mathbf{C}^n$ contains all claims in $L^{\infty}$
which  can be superreplicated. We define the set $\mathbf{M}_e^n$ of equivalent local martingale measures for the finite market $n$ as
\begin{align}\label{Me}
 \mathbf{M}_e^n &=
 \{Q\sim P|_{\cF_{T^*}}: \mathbf S^n \text{ is local $Q$-martingale}\}
 \end{align}
Besides this, we assume that for each finite market $n$ no arbitrage holds, i.e.
\begin{equation}\label{emm}
\mathbf{M}^n_e\ne\emptyset,\quad\quad\text{ for all $n\in\N$}.
\end{equation}
However, there is still the possibility of approximating an arbitrage profit by trading on the sequence of market models which motivates the notion of \emph{no asymptotic free lunch}.

\begin{definition}\label{N(A)FL} A given large financial market satisfies NAFL if
$$
\overline{\bigcup_{n=1}^{\infty}\mathbf{C}^n}^*\cap L^{\infty}_+ =\{0\}
$$ where $\overline{\mathbf{C}}^*$ denotes the closure of $\mathbf{C} \subset L^\infty$ with respect to the weak-$^*$ topology.
\end{definition}

\begin{definition}\label{NAFL}
The term structure model 	$ \{(P(t,T))_{0 \le t \le T}: 0 \le T \le T^*\}$ satisfies NAFL if there exists a dense sequence
$(T_i)_{i\in\N}$ in $[0,T^*]$, such that  the large financial market of Definition~\ref{LFM} satisfies the condition NAFL.
\end{definition}

Theorem 5.2 in  \cite{KleinSchmidtTeichmann2015} shows that NAFL is equivalent to existence of an equivalent local martingale measure (ELMM) which we recall here for convenience.
\begin{theorem}\label{th1}
Assume that Assumptions \ref{ass1} and \eqref{emm} hold.
The family of term structure models $ \{(P(t,T))_{0 \le t \le T}: 0 \le T \le T^*\}$  satisfies NAFL, if and only if there exists a measure $Q^*\sim P|_{\cF_{T^*}}$ such that
\begin{align}\label{EMM}
( (X_t^0)^{-1} P(t,T))_{0 \leq t \leq T} \ \text{are local $Q^*$-martingales for all } T \in [0,T^*].
\end{align}
\end{theorem}

\subsection{An extension of the HJM-approach}\label{sec:HJM}

Now we are in the position to extend the  HJM approach in an appropriate way to obtain arbitrage-free defaultbale term structure models under weak assumptions. Consider a measure $Q^*\sim P$. Our intention is to find conditions which render $Q^*$ an equivalent local martingale measure. 
From now on, only occasionally the measure $P$ will be used, such that all appearing terms (like martingales, almost sure properties, etc.) are to be considered with respect to $Q^*$ if not stated otherwise.

\subsubsection*{Announced times}  
The default indicator process $H$ is a bounded, c\'adl\'ag and increasing process, hence a submartingale of class (D). By the Doob-Meyer decomposition, the process 
\begin{align}
\label{eq:M} M_t = H_t - H^p_t, \quad t \ge 0 
\end{align}
is a true martingale where $H^p$ denotes the dual $\bbF$-predictable projection, also called  compensator, of $H$. As $1$ is an absorbing state, $H^p_t=H^p_{t \wedge \tau}$. The construction of random times satisfying certain specifications can be approached from different viewpoints while, to the best of our knowledge, an existence result covering our general setting  is not yet available in the literature. However, in Section \ref{sec:doublystochastic} we provide an existence  under when immersion holds following the ideas from  \cite{Jacod1975}. For more general results and related literature,  see  \cite{LiRutkowski2012} and \cite{Song2014}.

To keep the arising technical difficulties at a minimum, we assume that $H^p$ can be decomposed in an absolutely continuous and a (predictable) pure-jump part, such that
\begin{align}\label{Hp}
H^p_t = \int_0^{t \wedge \tau} h_s ds + \int_0^{t \wedge \tau} \int_\R x \Gamma(ds,dx), \quad t \ge 0,
\end{align}
with a non-negative process $h$ and with a predictable integer-valued random measure satisfying $\Gamma(dt,dx)=\sum_{s >0}  \ind{\Delta H^p_s > 0}\delta_{(s,\Delta H^p_s)} (dt,dx)$; here $\delta_x$ denotes the Dirac measure at the point  $x$.  
The process $h$ can be chosen predictable, see Theorem 2.1  in \cite{DelbaenSchachermayer1995}. 

Whenever $\Delta H^p_\sigma > 0$, for a predictable time $\sigma$, there is a positive probability that the company defaults at time $\sigma$ and the probability of this event is related to $\Delta H^p_\sigma$. We call such times \emph{risky times}, i.e.\ predictable times having a positive probability of a default occuring right at that time. 

A random set $A$ is called \emph{thin}, if $A=\bigcup \llbracket A_n \rrbracket$ with some stopping times $A_n$ and the associated stochastic intervals $\llbracket A_n \rrbracket=\{(\omega,t):t \in \R_{\ge 0}, t=A_n(\omega)\}$. The set of risky times is a thin set  $\bigcup_{i \ge 1}\llbracket U_i \rrbracket $ with predictable times $U_1,U_2,\dots$. These times may be chosen in such a way that $\llbracket U_i \rrbracket \cap \llbracket U_j \rrbracket = \emptyset$ if $j\neq i$. For our purposes it will be convenient to work with the following  class of predictable times.

\begin{definition}
We call a random time $U$ \emph{announced} if there exists an $\bbF$-stopping time $S$ with $S<U$ almost surely and $U$ is $\cF_S$-measurable.  
\end{definition}
The intuition behind this definition is as follows: at the announcement time $S$ the market receives new information about a future date $U$ (i.e.\ $S<U$) at which default may happen with positive probability. For example, at time $S$ the market realizes that a country has difficulties to pay some of its obligations which are due at the coupon payment date $U$. See Example \ref{ex:announcingtimes} and  Section \ref{sec:filtering} for applications under incomplete information.  Note that any deterministic, positive time is announced and that an announced time is always predictable. The concept of an announced time has similarities with the concept of an \emph{announcing sequence} of the predictable time $U$ (see Theorem I.2.15 in \cite{JacodShiryaev}), which is a sequence of optional times $(S_n)$  being strictly smaller than $U$ and increasing to $U$.  

To ensure that the subsequent analysis is meaningful, we make the  following technical assumptions.
\begin{itemize}
\item[{\bf (A1)}] The process $h$ is non-negative, predictable and integrable on $[0,T^*]$:
	$$ \int_0^{T^*} |h_s| ds < \infty, \quad Q^*\text{-a.s.}, $$
\item[{\bf (A2)}]  the random measure $\Gamma(ds,dx)$ is given by
	$$ \Gamma([0,t],dx) = \sum_{i=1}^N \ind{U_i \le t} \delta_{\Gamma_i}(dx), $$
	where each risky time $U_i$ is announced, say by $S_i$, and
	$\Gamma_i:\Omega \to (0,1)$ is $\cF_{S_i}$-measurable, $1 \le i \le N$.
\end{itemize}
 Assumption {\bf (A2)} implies that the set $\cU$ of (default) risky times is finite. This is a reasonable assumption while working on a finite time interval. If $\Gamma_i=1$, default happens with probability one at time $U_i$, a case which we exclude for simplicity of exposition. 

\subsubsection*{Dynamic term structures with discontinuities}
Regarding defaultable bond prices we will start from a forward-rate framework and allow for discontinuities  in the term structure at  (default) risky times.  
Consider current time $t\in [0,T^*)$ and a bond with maturity $T\in (t, T^*]$. If the risky time  $U_i$ was announced before time $t$, investors will obtain an additional premium for the event $\{\tau = U_i\}$ only when $T \ge U_i $. For $T<U_i$ the investors are not exposed to this risk and hence will not receive an additional premium. This naturally leads to a discontinuity in the term structure $T \mapsto P(.,T)$ at $U_i$. 
Motivated by this, we  consider  a family of random measures $(\mu_t)_{t \ge 0}$, defined by
	$$ \mu_t(du) :=  \sum_{S_i \le t} \delta_{U_i}(du) $$
and assume that defaultable bond prices are given by
	\begin{align}\label{PtT}
	P(t,T) = \ind{\tau>t} \exp\bigg( -\int_t^T f(t,u) du - \int_t^T g(t,u) \mu_t(du)\bigg), \quad 0 \le t \le T \le T^*.
	\end{align}
Note that the right-continuity in Assumption \ref{ass1} naturally holds under \eqref{PtT}.  
The processes $f$ and $g$ are assumed to be It\^o processes of the form
\begin{align}\label{f1}
f(t,T) &= f(0,T) + \int_0^t a(s,T)ds + \int_0^t b(s,T) \cdot dW_s,\\
g(t,T) &= g(0,T) + \int_0^t \alpha(s,T)ds + \int_0^t \beta(s,T) \cdot dW_s,  \label{g1}
\end{align}
with an $n$-dimensional $Q^*$-Brownian motion $W$. By $\cB$ we denote the Borel $\sigma$-algebra generated by the open sets in $\R_{\ge 0}$ and by $\cO$ we denote the optional $\sigma$-algebra generated by all $\bbF$-adapted c\`adl\`ag processes. We will need the following technical assumptions.
\begin{itemize}
\item[{\bf (B1)}] the initial forward curves $f(\omega,0,t)$ and $g(\omega,0,t)$ are $\cF_0 \otimes \cB$-measurable, and integrable on $[0,T^*]$:
    $$\int_0^{T^*}|f(0,u)|+ |g(0,u)|du<\infty, \qquad Q^*\text{-a.s.},$$
\item[{\bf (B2)}] the \emph{drift parameters} $a(\omega,s,t)$ and $\alpha(\omega,s,t)$ are $\mathbb{R}$-valued, and $\mathcal{O}\otimes\mathcal{B}$-measurable. The parameter $a$ is integrable on $[0,T^*]$:
    $$\int_0^{T^*}\int_0^{T^*} |a(s,t)| ds\,dt<\infty, \quad Q^*\text{-a.s.},$$
    while $\alpha$ is bounded on $[0,T^*]$:
    $$ \sup_{s,t\leq T^*} |\alpha(s,t)| <\infty, \quad Q^*\text{-a.s.},$$
\item[{\bf (B3)}] the \emph{volatility parameter} $b(\omega,s,t)$ is $\mathbb{R}^n$-valued, $\mathcal{O}\otimes\mathcal{B}$-measurable, and bounded on $[0,T^*]$:
    $$ \sup_{s,t\leq T^*} \parallel b(s,t) \parallel <\infty, \quad Q^*\text{-a.s.,}$$
while $\beta(\omega,s,t)$ is $\R^n$-valued, $\cO\otimes\cB$-measurable, and  square integrable on $[0,T^*]$:
    $$\int_0^{T^*}\int_0^{T^*} \parallel\beta(s,t)\parallel^2  ds\,dt<\infty, \quad Q^*\text{-a.s.,}$$
\item[{\bf (B4)}] we assume that the dual predictable projection $\nu$ of the integer-valued random measure $\mu(dt,du)=\sum_{i = 1}^n \delta_{(S_i,U_i)}(dt,du)$ satisfies  $\nu(dt,du) = \nu(t,du)dt$ with a kernel $\nu(\omega,t,du)$, and
$$ \int_0^{T^*} \int_0^{T^*}|e^{-g(t,u)}-1| \, \nu(t,du)dt < \infty, \quad Q^*\text{-a.s.} $$
	Moreover $Q^*(\tau=S_i)=0$ for all $i \ge 1$.
\end{itemize}
\begin{remark}
Assumption (B4) requires that announcing times are totally inaccessible, i.e.\ come as a surprise. Moreover, there is no default by news, i.e.\ $\tau$ does not coincide with an announcing time. Both assumptions have been made to simplify the exposition but could be relaxed without big difficulties at the cost of lengthier formulas. For details on the compensator or the dual predictable projection of a random measure we refer to section II.1.1a in \cite{JacodShiryaev}.
\end{remark}

The following result gives the desired drift condition rendering the considered measure $Q^*$ an equivalent local martingale measure.  In this case, the family of term structure models $\{P(t,T)_{0 \le t \le T}: 0 \le T \le T^*\}$ satisfies NAFL by Theorem \ref{th1}.  
Set
\begin{align*}
\bar{a}(t,T) &= \int_t^T a(t,u) du, \\
\bar{b}(t,T) &= \int_t^T b(t,u) du, \\
\bar{\alpha}(t,T)&=\int_t^T \alpha(t,u)\mu_t(du), \\
\bar{\beta}(t,T)&=\int_t^T \beta(t,u)\mu_t(du).
\end{align*}

\begin{theorem}\label{thm:dc}
Assume that {\bf (A1)}-{\bf (A2)} and {\bf (B1)}-{\bf (B4)} hold. Then $Q^*$ is an ELMM if and only if the following two conditions hold: 
\begin{align}
&\int_0^t f(s,s)ds + \sum_{U_i \le t} g(U_i,U_i)  = \int_0^t (r_s + h_s) ds - \sum_{U_i \le t} \log (1-\Gamma_i) , \label{dc1} \\
&\bar a(t,T) +\bar \alpha(t,T) =  \half \parallel \bar b(t,T) + \bar \beta(t,T)\parallel^2 
+ \int_t^T\left(e^{-g(t,u)}-1\right)\nu(t,du), \label{dc2}
\end{align}
$0 \le t \le T\le T^*$, $dQ^* \otimes dt$-almost surely on $\{t < \tau \}$. 
\end{theorem}

In comparison to the classical  HJM drift condition  in the default-risk free case, $\bar a(t,T)=\half  \parallel \bar b(t,T) \parallel^2  $, a number of additional terms appear here. First, Equation \eqref{dc1} under $g(.,.)=0$ and $\cU=\emptyset$ is a well-known condition in intensity-based dynamic term structure models and relates the short rate accumulated by the bond, $f(t,t)$, with the risk-free short rate plus a compensation for default risk.  
The additional terms incorporate additional returns due to the extra default risk at risky times. These terms appear, to the best of our knowledge, for the first time in defaultable HJM-models. Technically, they originate from joint jumps in $1-H$ and its compensator. It turns out, that if $\Delta H^p \neq 0$, then  a classical HJM-approach with $g(.,.)=0$ allows for arbitrage profits.

The additional term in  \eqref{dc2}, $\int_t^T\left(e^{-g(t,u)}-1\right)\nu(t,du)$, appears as compensation for jumps in the term structure at news  news arrival times $S_1,S_2,\dots$ and can be linked to similar expressions in classical HJM-Models with jumps as for example in \cite{eberlein-oezkan-03}.

The following simple example illustrates the extension of our approach over intensity-based models and builds up intuition on condition \eqref{dc1}. A prominent representative is the Merton model, discussed in Example \ref{exp:merton_model}. In Section \ref{sec:filtering}, we also provide evidence that in models with incomplete information the proposed framework applies. For practical applications we will develop piecewise stochastic continuous affine models in Section \ref{secAffine}.

\begin{example}\label{ex:doublystochastic}
Consider a non-negative integrable and progressive process $\lambda$, constants $0<u_1< \dots < u_N$,  positive random variables $\lambda'_1,\dots,\lambda'_N$, with $\lambda'_i$ being $\cF_{u_i}$-measurable, and set
$$ \Lambda_t = \int_0^t \lambda(s) ds + \sum_{u_i \le t} \lambda'_i. $$ 
Let $\zeta$ be a standard exponential random variable, independent from $\Lambda$, and  set 
$$ \tau = \inf\{t\ge 0: \Lambda_t \ge \zeta \}. $$ 
This is a so-called doubly stochastic model and many variants of this example have been successfully applied in credit risk (see \cite{BieleckiRutkowski2002} and \cite{JeanblancChesneyYor2009}, for example). In the cases previously studied in the literature, however, instead of constant times $u_i$, totally inaccessible stopping times were considered, such that $H^p$ turns out to be absolutely continuous (see \cite{SchererSchmidSchmidt12} for an example). Here, we have $\Delta H_{u_i}^p > 0$ because $u_i$ is a risky time: by the memoryless-property of exponential random variables,
\begin{align}\label{temp321} 
Q^*(\tau=u_i|\tau \ge u_i) = Q^*(\lambda'_i \ge \zeta)=\E^*[1-\exp(-\lambda'_i)]. 
\end{align}
If $\Lambda$ is {deterministic} and the short-rate vanishes, we obtain the following term-structure
\begin{align*} 
	P(t,T) &= \ind{\tau > t} Q^*(\tau >T | \tau > t ) 
	= \ind{\tau > t} \exp\Big( - \int_t^T \lambda(s) ds - \sum_{u_i \in (t,T]}\lambda'_i\Big) ,
\end{align*}
which clearly falls into the class of models considered here. 
A simple computation yields 
\begin{align} 
H_t^p = \int_0^{t \wedge \tau} \lambda(s) ds + \sum_{i:u_i \le (t \wedge \tau)} (1-e^{-\lambda'_i})
\end{align}
and it is easily checked that the drift conditions \eqref{dc1}-\eqref{dc2} hold.\hfill $\diamond$
\end{example}

The proof of Theorem \ref{thm:dc} will make use of the following lemma in which we derive the canonical decomposition of the second integral in \eqref{PtT}, denoted by
\begin{align}
\label{eq_dyn_G}
I(t,T) := \int_t^T g(t,u)\mu_t(du), \qquad 0\leq t\leq T.
\end{align}
\begin{lemma}\label{thelemma}
Assume that {\bf(A1)}, {\bf(A2)}, and {\bf(B1)}, {\bf(B2)} hold. Then, for each $T \in [0,T^*]$ the process
$(I(t,T))_{0 \le t \le T}$ is a special semimartingale and
\begin{align*}
I(t,T) &= \int_0^t\bar{\alpha}(s,T)ds + \int_0^t\bar{\beta}(s,T)\cdot dW_s +  \int_0^t\int_0^T g(s,u)\ind{s<u}\mu(ds,du) -\int_0^t  g(s,s) \mu^U(ds)
\end{align*}
with $\mu^U(ds)=\sum_{i=1}^n \delta_{U_i}(ds)$.
\end{lemma}
\begin{proof}
We start with the observation that, by the definition of $\mu_t$,
\begin{align}
\nonumber I(t,T) &= \int_0^t\int_t^T g(t,u)\mu(ds,du)\\
\nonumber &= \int_0^t\int_0^T\ind{u>t}g(t,u)\mu(ds,du)\\
\label{eq_dyn1}
&= \int_0^t\int_0^T\Ind_{[0,u)}(t)g(t,u)\mu(ds,du).
\end{align}
The semimartingales $(\Ind_{[0,u)}(t)g(t,u))$ have the following canonical decompositions,
\begin{align}\label{dyng}
\nonumber \Ind_{[0,u)}(t)g(t,u) &= g(0,u) + \int_0^t \Ind_{[0,u]}(v)\, dg(v,u)+ \int_0^t g(v,u)d(\Ind_{[0,u)}(v)) \\
&= g(0,u)  + \int_0^t\Ind_{[0,u]}(v)\alpha(v,u)dv + \int_0^t\Ind_{[0,u]}(v)\beta(v,u)\cdot dW_v - g(u,u)\ind{u\leq t} 
\end{align}
and we obtain that
\begin{align*}
\eqref{eq_dyn1} &= \int_0^t\int_0^T g(0,u) \mu(ds,du) + \int_0^t\int_0^T \int_0^t\Ind_{[0,u]}(v)\alpha(v,u)dv  \,\mu(ds,du) \\
& + \int_0^t\int_0^T \int_0^t\Ind_{[0,u]}(v)\beta(v,u) \cdot dW_v  \,\mu(ds,du) - \int_0^t\int_0^T g(u,u)\ind{u\leq t} \mu(ds,du) \\ & =: (1') + (2') + (3')+(4').
\end{align*}
With (B2) it is possible to interchange the appearing integrals as the integral with respect to $\mu$ is  a finite sum. Hence,
\begin{align*}
(2') &= \int_0^t\int_0^t \int_0^T \Ind_{[0,u]}(v)\alpha(v,u)  \mu(ds,du) dv \\
&= \int_0^t \int_0^v \int_0^T\Ind_{[0,u]}(v)\alpha(v,u)  \mu(ds,du) dv
+ \int_0^t \int_v^t \int_0^T \Ind_{[0,u]}(v)\alpha(v,u)  \mu(ds,du) dv \\
&= \int_0^t \int_0^v \int_0^T\Ind_{[0,u]}(v)\alpha(v,u)  \mu(ds,du) dv
+ \int_0^t\int_0^T \int_0^s\Ind_{[0,u]}(v)\alpha(v,u)  dv \, \mu(ds,du)
\end{align*}
with an analogous 
expression for $(3')$.
Note that the first term in the last line equals $\int_0^t \bar \alpha(v,T) dv$. By \eqref{dyng},
\begin{align*}
 \int_0^s\Ind_{[0,u]}(v)\alpha(v,u)  dv + \int_0^s\Ind_{[0,u]}(v)\beta(v,u) \cdot  dW_v &= \Ind_{[0,u)}(s)g(s,u) - g(0,u) +  g(u,u) \ind{u \le s}
\end{align*}
such  that \eqref{eq_dyn1} is equal to
\begin{align*}
  \int_0^t \bar \alpha(v,T) dv + \int_0^t \bar \beta(v,T)\cdot dW_v + \int_0^t\int_0^T \Ind_{[0,u)}(s) g(s,u)\mu(ds,du) -  \int_0^t\int_0^T \Ind_{[s,t]}(u) g(u,u) \mu(ds,du).
\end{align*}
By Assumption (A2),
\begin{align*}
\int_0^t\int_0^T \Ind_{[s,t]}(u) g(u,u) \mu(ds,du) = \sum_{U_i \le t} g(U_i,U_i) = \int_0^t g(s,s) \mu^U(ds) 
\end{align*}
which is a special semimartingale and we conclude.
\end{proof}

The previous lemma allows us to obtain the semimartingale representation of
$$ G(t,T):= \exp(-I(t,T)), \qquad 0 \le t \le T. $$
\begin{proposition}\label{prop43}
	Assume that {\bf(A1)}, {\bf(A2)} and {\bf(B1)}, {\bf(B2)} and {\bf (B4)} hold. Then,
	\begin{align*}
		\frac{dG(t,T)}{G(t-,T)} &= \bigg( -\bar{\alpha}(t,T)+ \half \parallel \bar \beta(t,T) \parallel ^2 + \int_t^T\left(e^{-g(t,u)}-1\right)\nu(t,du)\bigg) dt \\
		&- \bar \beta(t,T) \cdot dW_t  + \left(e^{g(t,t)}-1\right)\mu^U(dt) + dM^1_t,
	\end{align*}
with 
a local martingale $M^1$.
\end{proposition}
\begin{proof}
The It\^o-formula together the representation of $G$ given in Lemma \ref{thelemma} yields that 
\begin{align}
	G(t,T) &= G(0,T) + \int_0^t G(s-,T)\Big( - \bar{\alpha}(s,T)+ \half \parallel \bar \beta(s,t) \parallel^2 \Big)ds  - \int_0^t G(s-,T) \bar \beta (s,T) \cdot dW_s  \nonumber \\
	&+\int_0^t \int_0^T G(s-,T)\left(e^{-g(s,u)\Ind_{\{u>s\}}}-1\right)\mu(ds,du) 
	 + \int_0^t G(s-,T)\left(e^{g(s,s)}-1\right)\mu^U(ds).  \label{dynH}
\end{align}

Using Assumption (B4), we compensate $\mu(ds,du)$ by $\nu(s,du)ds$ and obtain the result.
\end{proof}

\begin{proof}[Proof of Theorem \ref{thm:dc}]
Set $F(t,T) := \exp\Big(-\int_t^T f(t,u) du \Big)$, $E(t):=\ind{\tau > t}$ such that
\begin{align*}
P(t,T) = E(t) F(t,T) G(t,T) .
\end{align*}
Then, by integration by parts,
\begin{align}\label{temp399}
dP(t,T) &=  F(t,T) G(t-,T) d E(t) + E(t-) d(F(t,T) G(t,T)) + d [E, F(.,T) G(.,T)]_t \\
& =: (1'')+(2'')+(3'') \nonumber
\end{align}
and we compute the according terms in the following. Regarding $(1'')$, we  obtain from  \eqref{Hp}, that
\begin{align}\label{defM2}
E(t) + \int_0^{t \wedge \tau } h_s ds
	 + \int_0^{t \wedge \tau}  \int_{\R} x \, \Gamma(ds,dx)=: M^2_t  
\end{align}
is a martingale.
Regarding (2$''$), we have that
\begin{align*}
d (F(t,T) G(t,T)) &= G(t-,T) d F(t,T)+  F(t,T) d G(t,T) + d\,\langle G^c(.,T), F^c(.,T) \rangle_t,
\end{align*}
where  $F^c(.,T)$ and $G^c(.,T)$ are the continuous local martingale parts of $F(.,T)$ and $G(.,T)$, respectively.
Computing the dynamics of $F(t,T)$ follows the original arguments of \cite{HJM}, see Lemma 6.1 in \cite{Filipovic2009},
such that for $0 \le t \le T$,
\begin{align}\label{FtT}
	d F(t,T) = F(t,T) \Big( f(t,t) - \bar a(t,T) + \half \parallel \bar b(t,T) \parallel^2 \Big) dt - F(t,T) \bar b(t,T) dW_t.
\end{align}
Together with Proposition \ref{prop43} this leads to
\begin{align}\label{temp419}
	\lefteqn{ \frac{d (F(t,T) G(t,T))}{F(t,T)G(t-,T)} = M_t^3 +\left(e^{g(t,t)}-1\right)\mu^U(dt) } \qquad  \nonumber\\
	&+\bigg( f(t,t) - \bar a(t,T) + \half \parallel \bar b(t,T) + \bar \beta(t,T)  \parallel^2   -\bar{\alpha}(t,T)  \bigg) dt  \\
	& - (\bar b(t,T) + \bar \beta(t,T)) \cdot  dW_t + \int_t^T\left(e^{-g(t,u)}-1\right)\nu(t,du)dt, \nonumber
\end{align}
where we used that  $\parallel \bar b(t,T) \parallel^2 + \parallel \bar \beta(t,T) \parallel^2 + 2 \bar b(t,T) \cdot \bar \beta(t,T)^\top = \parallel \bar b(t,T) + \bar \beta(t,T) \parallel ^2$ and a local martingale  $M^3$.
In view of (3$''$), we obtain from \eqref{dynH} that
\begin{align*}
	\frac{\Delta G(t,T)}{G(t-,T)} 
		& =  \int_t^T (e^{-g(t,u)}-1) \mu(\{t\}, du) + (e^{g(t,t)}-1)\mu^U(\{t\}).
\end{align*}
By Assumption (B4),  $\Delta E(t) \mu(\{t\},\R) = -\sum_{i\ge 1} \ind{\tau=t}\ind{S_i=t} = 0$. Hence, using \eqref{defM2},
\begin{align}\label{temp432}
	\begin{aligned}
	\sum_{0 < s \le t } \Delta E(s) \Delta G(s,T) &=  \int_0^t G(s-,T) (e^{g(s,s)}-1) \mu^U(\{s\}) d E(s) \\
	&= \int_0^t G(s-,T) (e^{g(s,s)}-1) \mu^U(\{s\}) d M^2_s \\
	&- \int_0^{t \wedge \tau}  \int_{\R}  G(s-,T) (e^{g(s,s)}-1) \mu^U(\{s\}) x \, \Gamma(ds,dx);
	\end{aligned}
\end{align}
where we used that for an integrable function $f:\R\to\R$, $\int f(s) \mu_T(\{s\}) ds = 0$ as $\mu$ is concentrated on a finite set. Note that    $(e^{g(t,t)}\mu^U(\{t\}))_{t \ge 0}$ is predictable due to Assumption (A2) and $\mu^U(\{s\}) \Gamma(ds,dx) = \Gamma(ds,dx)$.

Inserting \eqref{defM2}, \eqref{temp419} and \eqref{temp432} into \eqref{temp399}, we arrive that on $\{t < \tau\}$,
\begin{align*}
\frac{dP(t,T)}{P(t-,T)} &=  - h(t)  dt - \int_{\R} x \,\Gamma(dt,dx)\\
	 & +  \Big( f(t,t) +  \half \parallel \bar b(t,T) +\bar \beta(t,T) \parallel^2  - \bar a(t,T)  -\bar{\alpha}(t,T)\Big) dt \\
	 &+ \int_{\R}( e^{g(t,t)}-1)\mu^U(dt)\\
	 &+ \int_t^T\left(e^{-g(t,u)}-1\right)\nu(t,du)dt  \\
	 &- \int_{\R} (e^{g(t,t)} -1)x \, \Gamma(dt,dx)+ dM^4_t
\end{align*}
with a local martingale $M^4$. The process $(X_t^{-1}P(t,T))_{0 \le t \le T}$ is a local martingale if and only if the predictable part in the semimartingale decomposition vanishes. Letting $t=T$ one recovers
\begin{align*}
0 &= \int_0^t (f(s,s)-h(s)-r_s ) ds + \sum_{i:U_i \le t} \Big( e^{-g(U_i,U_i)}-1-\Gamma_ie^{g(U_i,U_i)} \Big)  
\end{align*}
for $0 \le t \le T^*$, on $\{t<\tau\}$, which is equivalent to $f(s,s)=h(s)+r_s $ and
\begin{align*}
1-e^{-g(U_i,U_i)} = \Gamma_i
\end{align*}
on $\{U_i \le T^* \wedge \tau\}$ such that \eqref{dc1} and \eqref{dc2} follow. The converse is easy to see. 
\end{proof}

\begin{example}[Announced random times]\label{ex:announcingtimes}
Consider a Poisson process  with intensity $1$ whose first $N$ jumping times $S_1<S_2< \dots < S_N$ denote the arrival times of news. There is a independent sequence $(\sigma_i)_{i \ge 1}$ of positive random variables with distribution function $F_\sigma$ and set
$$ U_i := S_i + \sigma_i. $$
Then, $U_i$ are announced by $S_i$ and we are just in a setting suggested by {\bf (B4)}. Assume for simplicity that $F_\sigma(x)=1-e^{-x}$, i.e.\ $\sigma_1$ is standard exponentially distributed and let 
$$ \tau = \inf\{t \ge 0:  t + \sum_{U_i \le t} 1\ge \Theta\} $$
with a standard exponential random variable $\Theta$, independent of all other appearing random variables. Then there is no deterministic risky time, i.e. $Q^*(\tau=t)=0$ for all $t \ge 0$. However,  each $U_i$ is a risky time because
$$ Q^*(\tau = U_i | S_i, \sigma_i, \tau \ge U_i) = 1-e^{-1}, $$ 
similar to   Equation \eqref{temp321}. 
\end{example}
A further example of an announced time will be studied in Section \ref{sec:filtering} under incomplete information.

\section{Generalized Merton Models}\label{sec:merton}

Inspired by  Merton's approach (see \cite{Merton1974} and Example \ref{exp:merton_model}) we assume that default of the company under consideration occurs when the company is not able to meet its liabilities. Payment dates of liabilities are considered deterministic, and hence payment dates turn out to be risky times in the sense of Section \ref{sec:HJM}. While in the previous section risky times were announced, but possibly random, in this section  we consider a deterministic and a finite  set $\cU:=\{u_1,\dots,u_N\}$ of times where default may occur with positive probability.  Denote
$$ \mu^M(du) = du + \sum_{i = 1}^N\delta_{u_i}(du). $$

Starting from \eqref{Hp}, we assume an appropriate modification of  {\bf (A2)}. 
\begin{itemize}
\item[{\bf (A2$^\prime$)}] 
The random measure $\Gamma(ds,ds)$ is given by  
$$ \Gamma([0,t],dx)=\sum_{i=1}^N \ind{u_i \le t} \delta_{\Gamma_i}(dx),$$  
with $\cF_{u_i-}$-measurable $\Gamma_i:\Omega \mapsto  (0,1)$, $i=1,\dots,N$.
\end{itemize}
A model satisfying {\bf (A2$^\prime$)} will be called \emph{generalized Merton model}. 
In this setting it is natural to consider $f=g$ directly, which we will do in the following. To this end, consider defaultable bond prices given by
	\begin{align}\label{PtTM}
	P^M(t,T) = \ind{\tau>t} \exp\bigg( -\int_t^T f(t,u) \mu^M(du)\bigg), \quad 0 \le t \le T \le T^*.
	\end{align}
	 We assume that  $f$ is an It\^o process and satisfies \eqref{f1}. Note that Assumption \ref{ass1} holds. 

Additionally, we require the following modifications of our previous assumptions:
\begin{itemize}
\item[{\bf (B1$^\prime$)}] the initial forward curve $f(\omega,0,t)$ is $\cF_0 \otimes \cB$-measurable, and integrable on $[0,T^*]$:
    $$\int_0^{T^*}|f(0,u)|<\infty, \qquad Q^*\text{-a.s.},$$
\item[{\bf (B2$^\prime$)}] the \emph{drift parameter} $a(\omega,s,t)$ is $\mathbb{R}$-valued $\mathcal{O}\otimes\mathcal{B}$-measurable and  integrable on $[0,T^*]$:
    $$\int_0^{T^*}\int_0^{T^*} |a(s,t)| ds\,dt<\infty, \quad Q^*\text{-a.s.},$$
\item[{\bf (B3$^\prime$)}] the \emph{volatility parameter} $b(\omega,s,t)$ is $\mathbb{R}^n$-valued, $\mathcal{O}\otimes\mathcal{B}$-measurable, and
    $$ \sup_{s,t\leq T^*} \parallel b(s,t) \parallel <\infty, \quad Q^*\text{-a.s.}$$
\end{itemize}

Moreover, we set
\begin{align}\label{abisbetabar}
\begin{aligned}
\bar{a}(t,T) &= \int_t^T a(t,u) \mu^M(du), \\
\bar{b}(t,T) &= \int_t^T b(t,u) \mu^M(du), \\
H'(t) &= \int_0^t h(s) ds - \sum_{i=1}^N \ind{u_i \le t } \log(1-\Gamma_i);
\end{aligned}
\end{align}
 note that in this section $\bar a$ has a different meaning from the previous section.
The following result, which is a direct corollary of Theorem \ref{thm:dc}, gives the desired drift condition in the generalized Merton models. 

\begin{corollary}\label{prop:dcm}
Assume  that {\bf (A1)},{\bf (A2$^\prime$)} and {\bf (B1$^\prime$)}-{\bf (B3$^\prime$)} hold. Then $Q^*$ is an ELMM if and only if the following two conditions hold:  
\begin{align}
\int_0^t f(s,s) \mu^M(ds) &= \int_0^t r_s ds + H'(t),  \label{dcm1}\\
\bar a(t,T)  &=  \half \parallel \bar b(t,T) \parallel^2 , \label{dcm2}
\end{align}
for $0 \le t \le T \le T^*$  $dQ^* \otimes dt$-almost surely on $\{t<\tau\}$.
\end{corollary}

\begin{example}[The Merton model]\label{exp:merton_model}
The paper \cite{Merton1974}  considers a simple capital structure of a firm, consisting only of equity and a zero-coupon bond with maturity $U>0$. The firm defaults at $U$ if the total market value of its assets is not sufficient to cover the liabilities. 

We are interested in setting up an arbitrage-free market for credit derivatives and consider a market of defaultable bonds $P^M(t,T)$, $0 \le t \le T \le T^*$ with $0< U \le T^*$ as basis for more complex derivatives. In a stylized form the Merton model can be represented by a Brownian motion $W$ denoting the normalized logarithm of the firm's assets, a constant $K>0$ and the default time
$$ \tau = \begin{cases}
U      & \text{if }W_U \le K \\
\infty & \text{otherwise}.
\end{cases}$$ 
Assume for simplicity constant interest rate $r$ and let $\bbF$ be the filtration generated by $W$.  Then $P^M(t,T) =e^{-r(T-t)}$ whenever $T<U$ because this bonds do not carry default risk. On the other hand, for $t< U \le T $, 
\begin{align}\label{Mertonbondprice}
P^M(t,T) &= e^{-r(T-t)}\E^*[\ind{\tau>T}|\cF_t]= e^{-r(T-t)} \E^*[\ind{\tau=\infty}|\cF_t]= e^{-r(T-t)}\Phi\bigg(\frac{W_t - K}{\sqrt{U-t}}\bigg),
\end{align}
where $\Phi$ denotes the cumulative distribution function of a standard normal random variable. For $t\to U$ we recover $P^M(U,U)=\ind{\tau=\infty}$. Note that this is indeed a generalized Merton model in the sense of the definition given on page \pageref{PtTM} and the derivation of representation \eqref{PtTM} with $\mu^M(du):=du+\delta_{U}(du)$ is straightforward. A simple calculation with 
\begin{align*}
P^M(t,T) &= \ind{\tau > t} \exp\bigg( -\int_t^T f(t,u) du - f(t,U)\ind {t < U \le T} \bigg)
\end{align*}
yields $f(t,T)=r$ for $T \not = U$ and 
$$ f(t,U) =- \log \Phi\bigg(\frac{W_t - K}{\sqrt{U-t}}\bigg). $$
By It\^{o}'s formula we obtain 
\begin{align*}
b(t,U)&= - \frac{\varphi\bigg(\frac{W_t - K}{\sqrt{U-t}}\bigg)}{\Phi\bigg(\frac{W_t - K}{\sqrt{U-t}}\bigg)}(U-t)^{-1/2},
\end{align*}
and indeed, $a(t,U)= \frac{1}{2}b^2(t,U).$ Note that  the conditions for Proposition \ref{prop:dcm} hold and, the market consisting of the bonds $P^M(t,T)$ satisfies NAFL, as expected. 

\end{example}
More flexible models of arbitrage-free bond prices can be obtained as in the following section and as in Section \ref{secAffine} on affine generalized Merton models.

\section{Structural Modelling under Incomplete Information}\label{sec:filtering}

Two approaches dominate the credit risk literature: the \emph{structural approach}, where a firm value is modelled together with a precise mechanism leading to default, and the \emph{reduced-form} approach. Reduced-form models proved very successful in calibration and pricing, and  the approach formulated in the first chapters of this work falls into this class. Typically, in reduced-form models  the compensator $H^p$ of the default indicator process $H$ is absolutely continuous (which coined the name \emph{intensity-based}). It was a remarkable insight of \cite{DuffieLando2001}, that introducing incomplete information in a structural model with predictable default time, as in \cite{BlackCox1976,LelandToft96} leads to an intensity-based model and hence connects these two approaches. 

In this section we study both approaches in more detail and show that structural models like the Merton-model under incomplete information on the firm's value do not lead to an intensity-based model: the discontinuities in $H^p$ persist even after projection. 

Besides this, inspired by incomplete information, we  give a constructive example of a default time where the Az\'ema supermartingale has upward and downward jumps in Section \ref{ex:Azema}. This default time complements the approaches in \cite{Song2013,LiRutkowski2012}, see also \cite{JeanblancSong2011a,JeanblancSong2011b,JiaoLi2014}, and sheds some light on the construction of default times beyond the so-called immersion property.  Remarkably, the framework proposed in Section \ref{sec:HJM} is general enough to cover even such cases.

Filtering and statistics always work under the objective probability measure  $P$, while reduced-form approaches always target an ELMM $Q^*\sim P$. In this section we will study structural models under the objective probability measure and motivate the existence of deterministic risky times. Recall that a risky time is a predictable time $S$ such that $P(\tau=S)>0$.  Such risky times persist under an equivalent change of measure. The converse was already observed in \cite[A.1]{ArtznerDelbaen95}, i.e.\ the existence of a default intensity is independent of an equivalent change of measure. Hence, in the context of reduced-form modelling it is necessary to consider a general approach, which is able to incorporate risky times (this is not the case if the default time avoids stopping times). Such a general approach was proposed and studied in Sections \ref{sec:HJM} and \ref{sec:merton}. 

We denote by  $V=(V_t)_{t\geq 0}$ the asset value process of the company under consideration and denote by $D=(D_t)_{t\ge 0}$ the liability process. The company defaults when the value of the firm's assets is not sufficient to meet the liabilities, i.e. 
\begin{align*}
\tau_D:=\inf\{t\ge 0:V_t-D_t < 0\}.
\end{align*}

\begin{example} Most structural approaches  fall in one of the  following classes:
\begin{enumeratei}
	\item In the Merton model, debt of size $K$ has to be repaid at time $U>0$, which in our setup can be covered by letting $D_t = \delta_{U}(t)K$. Extensions to more sophisticated capital structures have been proposed, amongst others, in \cite{Geske1977,GeskeJohnson84}. For example, consider constants $K_1,\dots,K_N$ representing obligatory payments of the company due at times $0<u_1<\dots<u_N$. The default occurs at the first time when a payment cannot be met which leads to $D_t = \sum_{i=1}^N  \delta_{u_i}(t) K_i$. 
	\item Other approaches, like for example \cite{LelandToft96} and \cite{DuffieLando2001} provide theoretical evidence, that it is optimal for the company's owners to liquidate the company if the firm value falls below some liquidation barrier $K'$ such that $D_t=K'$. Continuous, time-dependent barriers can be handled in a similar way, at the expense of tractability. For a detailed description we refer to \cite{SchmidtNovikov:2008}.   
\end{enumeratei}
\end{example}

\subsection{A filtering problem}
Incomplete information in credit risk has been considered in various approaches (see \cite{FreyRunggaldier2010} and \cite{FreySchmidt2011} for a guide to the literature). Here we give a precise formulation when continuous and discrete news are present by means of nonlinear filtering theory. 
We assume that the investors' information contains noisy observations of the firm’s asset value $V$ and some additional economic information $(I_n)_{n\ge 1}$ at random times $(J_n)_{n\ge 1}$. 
Let $\mu_I(dt,dz) = \sum_{n \ge 1} \delta_{(I_n,J_n)}(dt,dz)$ be the random measure associated with the marked point process $(I_n,J_n)_{n\ge 1}$. 

Assume for a moment discrete observations of the form
$$ Y'_n = A(V_{t_n}) + \xi_n $$
with a measurable function $A$, observation times $0<t_1<t_2<\dots$ and (for example, normal and i.i.d.) noise given by $\xi_1,\xi_2,\dots$. The classical filtering problem in discrete time consists of estimating the unobserved process $V$ given the observations. This problem is solved by computing the conditional distribution of $V$. In continuous time, a standard approach is to consider cumulated observations $Y_t = \sum_{n:t_n \le t} Y'_n$ such that the limit under appropriate scaling is given by $\int_0^t A(V_s) ds + W_t$ with a Brownian motion $W$. 

More generally, one models the information market participants have access to by the observation process $(Y_t)_{t\ge0}$,  given as the solution of the stochastic differential equation
\begin{align*}
     dY_t &=  A(t,V_t)dt + B(t,Y_t)dW_t + \int_{\R^+}f(t,z) \mu_I(dt,dz)
\end{align*}
with possibly random initial value $Y_0$.
Under the assumption that the dual predictable projection of the random measure $\mu_{I}(dt,dz)$ satisfies $\nu_I(t,dz)dt = \lambda_tF_I(dz)dt$ this filtering problem has been studied in \cite{CeciColaneri2013}, in the case where $V$ is a diffusion. More general filtering results can be obtained along the lines of \cite{Grigelionis72} and \cite{GrigelionisMikulevicius2011}, for example when the observation times $(J_n)$ are discrete (as in \cite{DuffieLando2001}).

Assume that $\bbF$ is the completed and right-continuous filtration generated by the observation $Y$ and default information, and zero interest rates for simplicity.  Then bond prices given as
$$ P(t,T)=\E^*\big[ \ind{\tau > T} | \cF_t\big] $$
under an equivalent measure $Q^*\sim P$ provide a market which satisfies NAFL.

Finite-dimensional solutions of the filtering problems are rare. One famous example is the so-called Kalman-Bucy filter, see \cite{LiptserShiryaev2001I} for a full account. In the following part, we consider a specification inspired by \cite{DuffieLando2001} in a Gaussian setup leading  to explicit solutions. 

\subsection{Merton model with unknown drift}\label{sec:MertonKalman1}
Assume that the firm's value $V$ is given as geometric Brownian motion with unknown drift, i.e.
$$ dV_t = V_t ( X dt + \sigma dW_t), \quad t \ge 0, $$
with $V_0=v>0$, $\sigma >0$ and a normally distributed random variable $X \sim \cN(\mu_X,\sigma_X)$. We assume that $X$ and $W$ are independent. 
The modelled corporation sold a bond with face value $K>0$ maturing at $T$. While the firm's value is not a traded asset, stocks are traded and, as observed in \cite{Merton1974}, can be viewed as a call option on the firm value. Most importantly, they can be used to estimate $V$ from observed prices.

We therefore assume that $V$ is observable and so is
$$ Y_t = Xt + \sigma W_t, \quad t \ge 0. $$
At the random time $S<T$ news  about a payment of size $K'$ due at time $U>T$ enters the market. This information comes with additional information on the drift $X$. We model this additional information by $Y'=X+\eta$, $\eta$ being a $\cN(0,\sigma_\eta)$-random variable, independent of all other appearing variables. The reason for this additional payment can be manifold, for example issuance of a junior debt as in \cite{GeskeJohnson84} or the purchase of new production capacities.

We follow \cite{GeskeJohnson84} and assume that  default occurs at $T$ or $U$ when the firm value is not sufficient to cover liabilities, such that  $D(t)=\delta_{T}(t)K + \delta_{U}(t)K'$ with positive constants $K,K'$. Default may happen either at $T$ or $U$, i.e.\
\begin{align}\label{eq:tauD}
\tau_D = \begin{cases} T, & V_T < K \\
U, & V_T \ge K \text{ and } V_U < K' \\
\infty & \text{otherwise.}
\end{cases}
\end{align}

The solution of the linear filtering problem before $S$ is well-known (see \cite{LiptserShiryaev2001I} Theorem 8.1). 
The conditional distribution of $X$ given the observation until time $t<S$ is normal and we denote its mean by $\hat X_t$ and (deterministic) variance by $\Sigma(t)$. Then $\Sigma$ solves the Riccati equation $d\Sigma(t)=-\Sigma(t)^2 dt$, such that $\Sigma(t)=\sigma_X(1+\sigma_X t)^{-1}$ with $\sigma_X=\sqrt{\Var(X)}$. Moreover,
\begin{align}\label{Xhat} 
	d\hat X_t = \frac{\Sigma(t)}{\sigma^2}\Big( dY_t - \hat X_t dt\Big),
\end{align}
with initial value $\hat X_0=\E[X].$ A simple calculation shows
$$\hat X_S := \E[X|Y_s:0 \le s \le S, Y'] = \hat X_{S-} + \frac{\Sigma(S)^{3/2}}{\Sigma(S)\cdot (\Sigma(S)+\sigma_\eta^2)} \big( Y' - \hat X_{S-}\big) $$
and $\Sigma(S)= \sigma_\eta^2 \Sigma(S-)\cdot \big( \Sigma(S-)+\sigma_\eta^2)^{-1}$. For $t>S$, one has a classical setting again, such that $\hat X$ satisfies \eqref{Xhat} with initial value $\hat X_S$ and $\Sigma(t) = \Sigma(S)(1+\Sigma(S) (t-S))^{-1}.$

 The available information in the market is denoted, as previously, by $\bbF$. It is the completed filtration generated by $Y$, and $(U,Y')\ind{S\ge \cdot}$. Note that $V$ is $\bbF$-adapted and hence $\tau$ is an $\bbF$-stopping time by \eqref{eq:tauD}. The following proposition gives the conditional default probability with $\Phi$ denoting the cumulative distribution function of the standard normal distribution.
\begin{proposition}\label{prop:merton1}
We obtain that for $0\le t < T$,
\begin{align}
	P(\tau =T | \cF_t) & = 
	\Phi\Big( \frac{  a'(t) - \hat X_t (T-t)}{b'(t)} \Big)
\end{align}
with the limit $P(\tau=T|\cF_T)=\ind{V_T<K}$ and
\begin{align*}
a'(t) &= \log (K \cdot V_t^{-1})+\half\sigma^2(T-t), \\
b'(t) &= \sqrt{T-t} \cdot \sqrt{ \sigma^2+ \Sigma(t) (T-t) }. 
\end{align*}
Moreover, for $S\le t<T$, 
\begin{align*} 
P(\tau =U | \cF_t) & = \Phi\Big( \frac{  a''(t) - \hat X_t (T-t) }{b''(t)} \Big),
\end{align*}
and for $T  \le t < U$ we have that
\begin{align}\label{temp720}
	P(\tau =U | \cF_t) & = \ind{\tau > T }\Phi\Big( \frac{  a''(t) - \hat X_t (U-t) }{b''(t)} \Big),
\end{align}
with limit $P(\tau=U|\cF_U)=\ind{V_T \ge K} \ind{V_U < K'}$ and
\begin{align*}
a''(t) &= \log (K' \cdot V_t^{-1}) +\half\sigma^2(U-t) , \\
b''(t) &= \sqrt{U-t}\cdot \sqrt{ \sigma^2+ \Sigma(t) (U-t)}. 
\end{align*}
\end{proposition}
\begin{proof}
For $0\le t < T$, we have 
\begin{align*}
	P(\tau =T | \cF_t) & = 
		P\Big( V_t \cdot e^{(X-\frac{\sigma^2}{2})(T-t)+\sigma(W_T-W_t)} \le K | \cF_t\Big) \\
		&= \E\Big[ \Phi\Big(\frac{\log(K/V_t)+(\frac{\sigma^2}{2}-X)(T-t)}{\sigma \sqrt{T-t}} \Big)|\cF_t\Big].
\end{align*}
For a $\cN(a,b^2)-$distributed random variable $\xi$, we obtain
\begin{align}
	\E[\Phi(\xi)] &= \int_{\R}\int_{y \le x} \frac{1}{\sqrt{2\pi}} e^{-\frac{y^2}{2}} dy \frac{1}{\sqrt{2\pi b^2}} e^{-\frac{(x-a)}{2b^2}} dx \nonumber\\
	&= P(\eta \le \xi) = P(\eta-\xi \le 0) = \Phi\Big(\frac{a}{\sqrt{1+b^2}}\Big), \label{eq:Ephi}
\end{align}
with $\eta$ being standard normal and independent of $\xi$. The first part of the proposition follows where the limit for $t\to T$ is easily verified. The second part follows in an analogous way. 
\end{proof}

As shown in Proposition \ref{prop:merton1}, default is predictable in this model such that $H_t=H^p_t$ with $\Delta H^p_t = 0$ for all $t \not \in \{T,U\}$. Moreover, both $U$ and $T$ are announced such that this model satisfies (A2) with $U_1=T$, $U_2=U$,  $\Gamma_1 = \ind{\tau=T}$ and $\Gamma_2 = \ind{\tau=U}$. A more sophisticated term structure would appear if the firm value followed a more complex model, for example with stochastic volatility. A highly tractable class in this regard is the class of affine models proposed in Section \ref{secAffine}.

It is straightforward to extend Proposition \ref{prop:merton1} to the case of more than one observation time $S$ or additional information as in \cite{DuffieLando2001}. The setting considered here is in some sense complementary to \cite{FreyRoeslerLu}, where the authors consider news on the firm value arriving at the times $S_1,S_2,\dots$, while mainly working under the assumption that these times are the jump times from a Poisson process and, whence, not predictable.

\begin{remark}[Bayesian estimation of the market price of risk]
In Merton's model, the Girsanov theorem yields that under any equivalent measure $Q'$ the firm value is again a geometric Brownian motion, but with possibly different drift $\mu'$. In contrast to Merton's original approach we do not assume that $V$ can be considered as traded asset such that typically, $\mu'$ does not equal the risk-free short rate $r$. Equity, considered as a call on the firm value, however gives access to a non-linear function of $\mu'$. Having observations from equity and firm-value at hand and assuming that the estimation of $\mu'$ contains additive noise, this leads to a filtering problem as above. 
The  developed  approach can now be used for a Bayesian estimation of $\mu'$ or, equivalently, the market price of risk. 

\end{remark}

\subsection{On the Az\'ema supermartingale}\label{ex:Azema}
Motivated by the above examples from filtering theory, we study the following, admittedly theoretical, example of a default time $\tau$. However, it clearly shows that our framework is more general than the one studied in \cite{BelangerShreveWong2004}, as in their equation (2.4) the Az\'ema supermartingale only jumps upward because the hazard process is non-decreasing.

 Consider an unobserved normal random variable $X$ taking the values $1$ and $2$, given on some fixed probability space $(\Omega,\cA,Q^*)$ with an equivalent measure $Q^*\sim P$. In addition, the observation $Y$ satisfies $Y_{t_i}=X+\xi_i$, $i\geq 1$, for i.i.d\,random variables $\xi_1,\xi_2,\dots$. In this section we consider the filtration $\bbG\subset \bbF$  as  the (augmented) filtration generated by the observations $Y_{t_1},Y_{t_2},\dots$. The market filtration $\bbF$ additionally contains information about $\tau$, but will not play a r\^ole here. Since $X$ is discrete the conditional distribution of $X$ given $\bbG$ can be computed in a direct manner. Assume that the default time is given by 
$$\tau=\inf\{t\ge 0| Xf(t)\ge E\},$$
where $f$ is a non-decreasing function tending to infinity and $E$ is a standard exponential random variable, independent of $X$. Then the \emph{Az\'{e}ma supermartingale} is  given by
$$Z_t = P(\tau>t|\cG_t) =E^*[\exp(-Xf(t))|\cG_t].$$
Denote by $\pi_t(A)= Q^*(X\in A|\cG_t)$ for $A\in\mathcal{B}(\R)$ the conditional distribution of $X$ given $\cG_t$. Then
$$E^*[\exp(-Xf(t))|\cG_t]=\int_{\R}\exp(-xf(t))\pi_t(dx).$$ 
This process jumps at every news update, i.e.\ at the times $t_1,t_2,\dots$. Obviously, if there is good news, we have an upward jump because the default probability decreases and vice versa. Hence the Az\'ema supermartingale typically has upward and downward jumps. This even holds if $f$ is continuous, and hence $Q^*(\tau=t)=0$ for all $t \ge 0$. 

In the same spirit on can replace the times $t_1,t_2,\dots$ with totally inaccessible times and will obtain again an Az\'ema supermartingale having upward and downward jumps.

\section{Existence of Arbitrage-free Models under Immersion}\label{sec:doublystochastic}

For calibrating the model it is important to identify assumptions, which specify a unique model which can be identified from available data. In this regard, we extend the results of Sections 5 and 6 in \cite{FilipovicSchmidtOverbeck2011} to the case where $H^p$ can have jumps. Then Theorem \ref{prop:dcm} (and, in a similar way, Theorem \ref{thm:dc}) states that under NAFL, the drift parameter $a(t,T)$  is determined by the volatility coefficient $b$. However, the first  drift condition gives an implicit relation of the risk-free short rate and the cumulative compensation for default risk, $H'$  to the process $(f(t,t))_{t \ge 0}$. The purpose of this section is to provide mathematical evidence that such arbitrage-free models exists uniquely under further assumptions. For simplicity, we concentrate on generalized Merton models with deterministic risky dates $u_1,\dots,u_N$.

A single point process is a point process with a single jump, say $\tau'$, and its path can be identified with $(\ind{t \ge \tau'})_{t \ge 0}$.
We therefore assume from now on that the stochastic basis satisfies the following structure. First, let
  \begin{enumerate}
  \item $(\Omega_1,\cG,(\cG_t)_{t \ge 0},Q_1)$ is some filtered probability
  space carrying the market information, in particular the Brownian
  motion $W(\omega)=W(\omega_1)$, the marked point process $(u_n,\Gamma_n)_{n \ge 1}$, the progressive process $h$, together with $\cG=\cG_\infty$.

  \item $(\Omega_2,\cH)$ is the canonical space of paths of single point processes endowed with the minimal filtration $(\cH_t)$:
  the generic $\omega_2\in\Omega_2$ is c\`ad, piecewise
  constant function from $\R_+$ to $\{0,1\}$ starting at $0$ and having at most one jump. Henceforth, we let
  \[H_t(\omega)=\omega_2(t).\] 
  The filtration $\bbH=(\cH_t)_{t \ge 0}$ is thus $\cH_t=\sigma(\tau \wedge s\mid s\le t)$, and
  $\cH=\cH_\infty$,

  \item $Q_2$ is a probability kernel from $(\Omega_1,\cG)$ to
  $\cH$ to be determined below.
  \end{enumerate}
 \begin{description}
\item[(A3)] $\Omega=\Omega_1\times\Omega_2$, $\cA=\cG\otimes\cH$,
$Q^*(d\omega)=Q_1(d\omega_1)Q_2(\omega_1,d\omega_2)$, where
  $\omega=(\omega_1,\omega_2)\in\Omega$, and
$\cF_t=\cG_t\otimes \cH_t$. 
\end{description}

\begin{theorem}\label{thm2gen}
Assume {\bf(A3)} holds. Let $f(0,T)$ and $b(t,T)$ satisfy {\bf(B1$^\prime$)}, and {\bf(B3$^\prime$)}, respectively. Define $a(t,T)$ in such a way that \eqref{dcm2} holds for all $(t,T)$.

Suppose, for any loss path $\omega_2\in\Omega_2$, there exist a family of adapted processes $(f(t,T)_{0 \le t \le T})$, $T \in (0,T^*]$, satisfying  \eqref{f1}, with  $(f(t,t))_{0 \le t \le T^*}$ being progressive and such that \eqref{dcm1} is satisfied. Then
\begin{enumeratei}
  \item\label{thm2gen1} {\bf{(B2$^\prime$)}} is satisfied.

  \item\label{thm2gen2} There exists a unique probability
kernel $Q_2$ from $(\Omega_1,\cG)$ to $\cH$, such that 
$$ H_t^p = \int_0^{t\wedge \tau} h(s) ds + \sum_{i=1}^N \ind{u_i \le t \wedge \tau} \Gamma_i $$
and the no-arbitrage condition \eqref{dcm2} holds.

\item\label{thm2gen4} $H^p$ is the compensator of $H$ with respect to $(\cG\otimes\cH_t)$. Moreover,
\begin{align}
  Q^*\left( \tau> t\mid   \cG \right)=Q_2\left(\tau> t \right)= e^{-\int_{0}^{t}h(s) ds} \cdot \prod_{u_i\le t} (1-\Gamma_i),\quad t\ge 0.\label{thm2gen4eq1}
\end{align}
\item\label{thm2gen3} $Q_2(\cdot, A)$ is $\cG_t$-measurable for all $A\in \cH_t$ and $t\ge 0$.
Consequently, every $\bbG$-martingale is an $\bbF$-martingale.

\end{enumeratei}
\end{theorem}
\begin{proof}
(i), (ii) and (iv) follow as in Theorem 5.1 in \cite{FilipovicSchmidtOverbeck2011}. Regarding (iii), note that
\begin{align*}
  \Psi(\omega_1,t)&:=Q^*(\tau>t\mid\cG)(\omega_1)
  = \E^*[1-H_t|\cG_\infty](\omega_1) \\
	&= 1- \int_0^t \Psi(\omega_1,s) h(\omega_1,s) ds - \sum_{u_i \le t} \Psi(\omega_1,u_i-) \Gamma_i(\omega_1) \\
	&= Q^*(\tau>t\mid\cG_t)(\omega_1),
\end{align*}
which implies \eqref{thm2gen4eq1}.
\end{proof}
Property (iv) of the theorem is known as ``(H)-hypothesis'', see \cite{BremaudYor1978,SchoenbucherEhlers,JeanblancRutkowski2000}. Note that in Example \ref{ex:Azema} the (H)-hypothesis does not hold. 
Formula \eqref{thm2gen4eq1} can be used for Monte-Carlo simulation, see Section 5.1 of \cite{FilipovicSchmidtOverbeck2011}.

\section{Affine generalized Merton models}\label{secAffine}
Affine processes are a well-known tool in the financial literature and one reason for this is their analytical tractability. They have been applied to a wide range of financial problems, in particular in term-structure modelling. The aim of this section is to give an affine specification of generalized Merton models. Affine processes  in the literature are assumed to be \emph{stochastically continuous} (see \cite{DuffieFilipovicSchachermayer} and \cite{Filipovic05}).  Due to the discontinuities in the term structure we propose to consider \emph{piecewise continuous affine processes}. To the best of our knowledge, this is the first time that such (time-inhomogeneous) affine processes are investigated.  

We assume that $ \cU=\{u_1,\dots,u_N\} $ and a vanishing short rate $r_t=0$ in this section and place ourselves in the context of Section \ref{sec:doublystochastic}. The idea is to consider an affine process $X$ and study arbitrage-free doubly stochastic term structure models where the compensator $H^p$ of the default indicator process $H=\ind{\cdot \le \tau}$ is given by
\begin{align} \label{affineHp} 
H_t^p = \int_0^t \Big( \phi_0(s)+\psi_0(s)^\top \cdot X_s \Big) ds + \sum_{i=1}^n\ind{t \ge u_i} \Big(1-e^{- \phi_i - \psi_i^\top \cdot X_{u_i}}\Big).
\end{align}
To ensure that $H^p$ is non-decreasing we will require that $\phi_0(s)+\psi_0(s)^\top \cdot X_s  \ge 0$ for all $s \ge 0$ and 
$\phi_i + \psi_i^\top \cdot X_{u_i} \ge 0$ for all $i=1,\dots,N$. 

We place ourselves in the doubly stochastic settting of Section \ref{sec:doublystochastic} and assume additionally that the probability space $(\Omega_1,\cG,(\cG_t)_{t \ge 0},Q_1)$ carries the $d$-dimensional Brownian motion $W$ and that $(\cG_t)_{t \ge 0}$ is generated by $W$ with the usual augmentation by null sets. There is no need of an additional point process here, as will become clear shortly.

In this regard, consider a state space in canonical form $\cX=\R_{\ge 0}^m \times \R^n$ for integers $m,n \ge 0$ with $m+n=d$ and a $d$-dimensional Brownian motion $W$. Let $\mu$ and $\sigma$ be defined on $\cX$ by
\begin{align}
\label{affine_mu}
\mu(x) &= \mu_0 + \sum_{i=1}^dx_i\mu_i,\\
\label{affine_sigma}
\half \sigma(x)^\top\sigma(x) &= \sigma_0 + \sum_{i=1}^dx_i\sigma_i,
\end{align}
where $\mu_0,\mu_i\in\R^d$, $\sigma_0,\sigma_i\in\R^{d\times d}$, for all $i\in\{1,\ldots,d\}$. We assume that the parameters $\mu^i,\ \sigma^i$, $i=0,\dots,d$ are admissible in the sense of Theorem 10.2 in \cite{Filipovic2009}. Then the continuous, unique strong solution  of the stochastic differential equation
\begin{align}
dX_t &= \mu(X_t)dt + \sigma(X_t)dW_t,\quad X_0=x,
\end{align}
is an \emph{affine} process $X$ on the state space $\cX$, see Chapter 10 in \cite{Filipovic2009} for a detailed exposition.

We call a bond-price model \emph{affine} if there exist functions $A:\R_{\ge 0}\times \R_{\ge 0} \to \R$, $B:\R_{\ge 0}\times \R_{\ge 0} \to \R^d$ such that
\begin{align}\label{affineMertonmodel}
P^M(t,T)= \ind{\tau > t} e^{-A(t,T)-B(t,T)^\top \cdot X_t},
\end{align}
for $0 \le t \le T \le T^*$. In view of Assumption \ref{ass1} we assume that $A(.,T)$ and $B(.,T)$ are right-continuous. Moreover, we assume that  $t \mapsto A(t,.)$ and $t \mapsto B(t,.)$ are differentiable from the right and denote by $\partial_t^+$ the right derivative.  
The following proposition gives sufficient conditions such that an affine generalized Merton model is arbitrage-free. 

\begin{proposition}\label{prop:affine}
Assume that $\phi_0:\R_{\ge 0}\to \R$, $\psi_0:\R_{\ge 0}\to \R^d$ are continuous,  $\psi_0(s)+\psi_0(s)^\top \cdot x \ge 0$ for all $s \ge 0$ and $x \in \cX$ and the constants $\phi_1,\dots,\phi_n \in \R$ and $\psi_1,\dots,\psi_n \in \R^d$ satisfy $\phi_i + \psi_i^\top \cdot x \ge 0$ for all $1 \le i \le n$ and $x \in \cX$.
Moreover, let the functions $A:\R_{\ge 0}\times \R_{\ge 0} \to \R$ and $B:\R_{\ge 0}\times \R_{\ge 0}\to \R^d$ be the unique solutions of 
\begin{align} \label{ric1}
\begin{aligned}A(T,T) &= 0 \\
A(u_i,T) &= A(u_i-,T)-\phi_i \\
- \partial_t^+ A(t,T) &= \phi_0(t) +  \mu_0^\top \cdot B(t,T) 
- B(t,T)^\top \cdot \sigma_0 \cdot B(t,T),  \\
\end{aligned} \intertext{and} \label{ric2}
\begin{aligned}
B(T,T) &= 0 \\
B_k(u_i,T) &= B_k(u_i-,T) - \psi_{i,k} \\ 
-\partial_t^+ B_k(t,T) &= \psi_{0,k}(t) +  \mu_k^\top \cdot B(t,T)  - B(t,T) ^\top \cdot \sigma_k \cdot B(t,T),
\end{aligned}
\end{align}
for $0 \le t \le T$.
Then, the doubly-stochastic affine model given by \eqref{affineHp} and \eqref{affineMertonmodel} satisfies NAFL. 
\end{proposition}
\begin{proof}
By construction,
\begin{align*}
A(t,T)  = \int_t^T a'(t,u) du+\sum_{i:u_i \in (t,T]} \phi_i \\
B(t,T)  = \int_t^T b'(t,u) du+\sum_{i:u_i \in (t,T]} \psi_i
\end{align*}
with suitable functions $a'$ and $b'$ and $a'(t,t)=\phi_0(t)$ as well as $b'(t,t)=\psi_0(t)$.
Comparison of \eqref{affineMertonmodel} with \eqref{PtTM} yields the following: on the one hand, for $T =u_i\in \cU$, we obtain $f(t,u_i)=\phi_i+\psi_i^\top \cdot X_t$.  Hence, the coefficients $a(t,T)$ and $b(t,T)$ in \eqref{f1} for $T=u_i \in \cU$ compute to $a(t,u_i)=\psi_i^\top \cdot \mu(X_t)$ and $b(t,u_i) = \psi_i^\top \cdot \sigma(X_t)$.

On the other hand, for $T \not \in \cU$ we obtain that $f(t,T) = a'(t,T)+b'(t,T)^\top \cdot X_t$. Then, the coefficients $a(t,T)$ and $b(t,T)$  can be computed by It\^o's formula, i.e.
\begin{align} \label{abaraffine}
\begin{aligned}
	a(t,T) &= \partial_t a'(t,T) + \partial_t b'(t,T)^\top \cdot X_t + b'(t,T)^\top \cdot \mu(X_t) \\
	b(t,T) &= b'(t,T)^\top \cdot \sigma(X_t).
\end{aligned}
\end{align}
Set $\bar a'(t,T)=\int_t^T a'(t,u) du$ and $\bar b'(t,T)=\int_t^T b'(t,u) du$ and note that, 

$$ \int_t^T \partial_t a'(t,u) du = \partial_t \bar a'(t,T) + a'(t,t). $$
As $\partial_t^+ A(t,T)=\partial_t \bar a'(t,T)$, and $\partial_t^+ B(t,T)=\partial_t \bar b'(t,T)$, we obtain from \eqref{abaraffine} that 
\begin{align*}
  \bar a(t,T) & = \int_t^T a (t,u) \mu^M (du) = \int_t^T a(t,u) du  + \sum_{u_i \in (t,T]} \psi_i^\top \cdot \mu(X_t) \\
  &= \partial_t^+ A(t,T) + a'(t,t) 
  + \big( \partial_t^+ B(t,T) + b'(t,t) \big)^\top \cdot X_t + B(t,T)^\top  \cdot \mu(X_t), \\ 
   \bar b(t,T) & =\int_t^T b(t,u) \mu^M(du) = \int_t^T b(t,u) du +  \sum_{u_i \in (t,T]} \psi_i^\top \cdot \sigma(X_t) \\
   &=B(t,T)^\top \cdot \sigma(X_t)
\end{align*}
for $0 \le t \le T \le T^*$.
We now show that under our assumptions, the drift conditions \eqref{dcm1}-\eqref{dcm2} hold: first, Assumption \ref{ass1} is satisfied.
Observe that, by equations \eqref{ric1}, \eqref{ric2}, and the affine specification \eqref{affine_mu}, and \eqref{affine_sigma}, 
the drift condition  \eqref{dcm2} holds. Moreover,
$$ \Delta H'(u_i) =  \phi_i + \psi_i^\top \cdot X_{u_i} $$
and $h(s)= \phi_0(s) + \psi_0(s)^\top \cdot X_s$ by \eqref{affineHp}. We recover $\Delta H_{u_i}^p = \Gamma_i=1-\exp(-\phi_i - \psi_i^\top \cdot X_{u_i})$ takes values in $[0,1)$ by assumption. Hence, 
 \eqref{dcm1} holds and the claim follows. 
\end{proof}

\begin{example}
In the one-dimensional case we consider  $X$,  given as solution of 
\begin{align*}
dX_t &= (\mu_0+\mu_1 X_t)dt + \sigma\sqrt{X_t}dW_t, \qquad t \ge 0.
\end{align*}
We assume for simplicity that $u_1=1$ and $N=1$ such that there is a single risky time, 1, and choose  $\mu^M(du)=\delta_{1}(du)$.  Moreover, let $\phi_0 = 0$, $\psi_0=1$ as well as $\phi_1=0$ and $\psi_1\ge 0$, such that
$$ H^p = \int_0^t X_s ds + \ind{t \ge 1} (1-e^{-\psi_1 X_1}). $$
Hence the probability of having no default at time $1$ just prior to $1$ is given by $e^{-\psi_1 X_1}$, compare Example \ref{ex:doublystochastic}.

An arbitrage-free model can be obtained by choosing $A$ and $B$ according to Proposition \ref{prop:affine} which can be immediately achieved using Lemma 10.12 from \cite{Filipovic2009} (see in particular Section 10.3.2.2 on the CIR short-rate model): denote $\theta=\sqrt{\mu_1^2+2\sigma^2}$ and
\begin{align*}
L_1(t) &= 2(e^{\theta t}-1) ,\\
L_2(t) &= \theta(e^{\theta t}+1) + \mu_1 (e^{\theta t}-1) ,\\
L_3(t) &= \theta(e^{\theta t}+1) - \mu_1 (e^{\theta t}-1) ,\\
L_4(t) &= \sigma^2(e^{\theta t}-1).
\end{align*}
Then 
\begin{align*} 
	A_0(s) &= \frac{2 \mu_0}{\sigma^2} \log \Big( \frac{2 \theta e^{\frac{(\sigma - \mu_1)t}{2}} } {L_3(t)}\Big), \quad B_0(s) = -\frac{L_1(t)}{L_3(t)} 
\end{align*}
are the unique solutions of the Riccati equations $B_0'=\sigma^2 B_0^2-\mu_1 B_0 $ with boundary condition $B_0(0)=0$ and $A_0'=-\mu_0 B_0$ with boundary condition $A_0(0)=0$. Note that with $A(t,T) = A_0(T-t)$ and $B(t,T)= B_0(T-t)$ for $0 \le t \le T < 1$, the conditions of Proposition \ref{prop:affine} hold. Similarly,
for $1 \le t \le T$, choosing
$A(t,T)=A_0(T-t)$ and $B(t,T)=B_0(T-t)$ implies again the validity of \eqref{ric1} and \eqref{ric2}. On the other hand, for $0 \le t <1$ and $T \ge 1$ we set $u(T)=B(1,T)+\psi_1=B_0(T-1)+\psi_1$, according to \eqref{ric2}, and let 
\begin{align*}
	A(t,T) &= \frac{2 \mu_0}{\sigma^2} \log \Big( \frac{2 \theta e^{\frac{(\sigma - \mu_1)(1-t)}{2}} } {L_3(1-t)-L_4(1-t)u(T)}\Big)\\
	B(t,T) &= -\frac{L_1(1-t)-L_2(1-t)u(T)}{L_3(1-t)-L_4(1-t)u(T)}.
\end{align*}
It is easy to see that \eqref{ric1} and \eqref{ric2} are also satisfied in this case, in particular  $\Delta A(1,T)=-\phi_1=0$ and $\Delta B(1,T)=-\psi_1$. Note that, while $X$ is continuous, the bond prices are not even stochastically continuous because they jump almost surely at $u_1=1$.  We conclude by Proposition \ref{prop:affine} that this affine model is arbitrage-free. \hfill $\diamond$
\end{example}

\section{Conclusion}\label{sec:conclusion} 

In this article we studied a new class of dynamic term structure models with credit risk where the compensator of the default time may jump at predictable times. This extends existing theory and allows including a number of structural approaches, like Merton's model, in a reduced-form model for pricing credit derivatives. More tractability can be obtained if the (default) risky times, i.e.\ predictable times where default can happen with positive probability, are even deterministic, which we studied in a class of generalized Merton models. Finally,  we provided a new class of highly tractable affine models which are  only piecewise stochastically continuous.

\end{document}